\documentclass[letterpaper, 10 pt, conference]{ieeeconf}

\IEEEoverridecommandlockouts
\usepackage{cite}

\usepackage{amsthm,amsmath,amssymb,amsfonts}
\usepackage{bm}
\usepackage{algorithm}
\usepackage{algpseudocode}

\usepackage{graphicx}
\usepackage{xcolor}

\usepackage{array}
\usepackage{multirow}
\usepackage[hyphens]{url}
\usepackage{hyperref}

\usepackage[normalem]{ulem}

\makeatletter
\let\MYcaption\@makecaption
\makeatother

\usepackage[font=footnotesize]{subcaption}

\makeatletter
\let\@makecaption\MYcaption
\makeatother

\usepackage{textcomp}
\def\BibTeX{{\rm B\kern-.05em{\sc i\kern-.025em b}\kern-.08em
T\kern-.1667em\lower.7ex\hbox{E}\kern-.125emX}}

\usepackage{lipsum}


\ifcsname theoremstyle\endcsname

\theoremstyle{plain}

\newtheorem{lemma}{Lemma}

\theoremstyle{definition}

\fi

\def\({\left(}
\def\){\right)}
\def\[{\left[}
\def\]{\right]}

\def\Phibf{{\bf \Phi}}
\def\Psibf{{\bf \Psi}}
\def\Lambdabf{{\bf \Lambda}}

\newif\ifshowWriterComment
\newcommand\writercomment[3]{\expandafter}

\ifcsname algorithmicrequire\endcsname

\fi

\def\alg#1{Algorithm~\ref{alg:#1}}
\def\fig#1{Fig.~\ref{fig:#1}}
\def\sec#1{Section~\ref{sec:#1}}
\def\eqn#1{\eqref{eqn:#1}}

\def\st{{\rm s.t.}}
\def\OptConsSep{&&\quad}

\makeatletter
\newcommand{\OptMin}{\@ifstar\OptMinStar\OptMinNoStar}
\makeatother
\newcommand{\OptMinStar}[3][]{
\ifx\\#1\\ \else\begin{subequations}\label{eqn:#1}\fi%
\begin{alignat}{2}
\min\ &\ #2 \nonumber \\
\st\ #3
\end{alignat}
\ifx\\#1\\ \else\end{subequations}\fi%
}
\newcommand{\OptMinNoStar}[3][]{
\ifx\\#1\\ \else\begin{subequations}\label{eqn:#1}\fi%
\begin{alignat}{2}
\min\ &\ #2 \ifx\\#1\\ \nonumber \else \tag{\ref{eqn:#1}}\label{eqnset:#1} \fi \\
\st\ #3
\end{alignat}
\ifx\\#1\\ \else\end{subequations}\fi%
}

\newcommand\OptCons[3]{
&\ #1
\ifx\\#2\\ \else \OptConsSep #2 \fi%
\ifx\\#3\\ \nonumber \else \label{eqn:#3} \fi%
}

\def\Tsim{T_\text{sim}}
\def\Drow{D_\text{row}}
\def\Dcol{D_\text{col}}
\renewcommand{\paragraph}[1]{\vspace*{0.25\baselineskip}\noindent\textbf{#1}}

\def\Phicomp{{\bf $\Phibf$-comp}}
\def\Psicomp{{\bf $\Psibf$-comp}}
\def\Lambdacomp{{\bf $\Lambdabf$-comp}}
\def\Conv{{\bf Conv}}

\title{\LARGE \bf Effective GPU Parallelization \\ of Distributed and Localized Model Predictive Control}

\author{Carmen Amo Alonso and Shih-Hao Tseng
\thanks{Carmen Amo Alonso and Shih-Hao Tseng are with the Division of Engineering and Applied Science, California Institute of Technology, Pasadena, CA 91125, USA.  Emails: {\tt\small \{camoalon,shtseng\}@caltech.edu}}
}

\showWriterCommenttrue

\begin{document}

\maketitle
\thispagestyle{empty}
\pagestyle{empty}

\bstctlcite{IEEE_BSTcontrol}

\begin{abstract}
To effectively control large-scale distributed systems online, model predictive control (MPC) has to swiftly solve the underlying high-dimensional optimization.
There are multiple techniques applied to accelerate the solving process in the literature, mainly attributed to software-based algorithmic advancements and hardware-assisted computation enhancements. However, those methods focus on arithmetic accelerations and overlook the benefits of the underlying system's structure.
In particular, the existing decoupled software-hardware algorithm design that naively parallelizes the arithmetic operations by the hardware does not tackle the hardware overheads such as CPU-GPU and thread-to-thread communications in a principled manner. Also, the advantages of parallelizable subproblem decomposition in distributed MPC are not well recognized and exploited.
As a result, we have not reached the full potential of hardware acceleration for MPC.

In this paper, we explore those opportunities by leveraging GPU to parallelize the distributed and localized MPC (DLMPC) algorithm. We exploit the locality constraints embedded in the DLMPC formulation to reduce the hardware-intrinsic communication overheads. Our parallel implementation achieves up to $50\times$ faster runtime than its CPU counterparts under various parameters.
Furthermore, we find that the locality-aware GPU parallelization could halve the optimization runtime comparing to the naive acceleration.
Overall, our results demonstrate the performance gains brought by software-hardware co-design with the information exchange structure in mind.

\end{abstract}
\section{Introduction}\label{sec:introduction}

The high computational demands of Model Predictive Control (MPC) have restricted its applicability to slow processes and low-dimensional systems \cite{qin2003survey}. Large-scale systems often require solving high-dimensional problems and algorithms scale poorly. Given the omnipresence of high-dimensional large-scale systems in real-world applications, efforts have been made to make MPC schemes more computationally efficient for these systems.

Multiple solutions have been proposed to accelerate MPC runtimes \cite{abughalieh2019survey}. One popular approach relies on providing computational enhancements to the optimization solving algorithms, either by exploiting the sparsity in the matrices or by finding initial points for the optimization \cite{jonson1984method,richter2009realtime,ke2017visual}. In this realm, recent works have also applied ideas from explicit MPC \cite{bemporad2002explicit} to large-scale networks \cite{oravec2017parallel,amoalonso2020explicit} by moving most of the computational burden outside the optimization algorithms. The other direction is to take advantage of state-of-the-art hardware such as multi-core processors (CPUs), many-core processors (GPUs) or field programmable arrays (FPGA) to perform computations in parallel \cite{kogel2012parallel,phung2017model,yu2017efficient,hyatt2020parameterized,plancher2019realtime,ling2008embedded,jerez2014embedded}. In some instances these two approaches are combined, so efficient optimization algorithms are solved using multiple threads via hardware-specific implementations.

Although these methods provide promising avenues, most of their efforts are centered around achieving efficient computations by appropriately exploiting algorithmic features, and rely on hardware to simply parallelize mathematical operations. Hence, the hardware implementation of the algorithms is completely decoupled from the original system formulation, and therefore any hardware-intrinsic overhead can only be handled by using efficient programming practices. Yet, some branches of MPC directly encoding parallelization features in their formulation have received very little attention in this field.

For instance, the merits of distributed MPC\footnote{Distributed MPC ports the ideas of standard MPC to the distributed setting, where different subsystems have different subcontrollers that operate in parallel and can communicate with each other in a local fashion \cite{scattolini2009architectures}.} have been overlooked in parallel settings \cite{abughalieh2019survey}, despite the fact that distributed MPC formulations are very well-suited for parallelization. Moreover, some distributed MPC frameworks allow the information exchange constraints among different subsystems \cite{amoalonso2020distributed}. This feature is not only present in most widespread large-scale systems \cite{doyle2005robust,negenborn2007multiagent,lynn2019physics,mahadevan2002dynamic}, but also these information exchange constraints resemble the hardware-intrinsic communication limitations and overheads encountered in MPC parallelization. Despite the great promise of these MPC frameworks to deal with parallelization and hardware-intrinsic overheads in a principled manner through the problem formulation, its full potential has not been realized in the literature.

In this paper we close this gap. We provide a principled parallel implementation and overhead analysis through an appropriate distributed MPC framework that allows for local communication constraints.  We exploit the potential for parallelization of this scheme in a GPU, where the GPU is not used to parallelize arithmetic computations but rather each computing thread is tasked with computing the operations corresponding to a subsystem in the network. Moreover, we demonstrate how simply applying standard parallelization techniques to the algorithm incurs unnecessary overhead. And we show that communication exchange constraints embedded in the framework allow us to explicitly deal with these hardware-intrinsic communication overheads in a principled manner by means of longest-vector length, combined kernels and local memory.

In particular, we take advantage of the recent work in \cite{amoalonso2020distributed} and its extension \cite{amoalonso2020explicit}, which provides a Distributed and Localized MPC (DLMPC) controller capable of encoding the communication structure of the network. It also provides a distributed algorithm to compute the DLMPC controller, where algorithmic iterations result in basic arithmetic operations that are scalable independently of the size of the network. We note that the limitations in communication among the GPU computing threads resemble the communication scheme in control systems for large-networks, and we take advantage of the local communication constraints that are already included in the DLMPC algorithm. We demonstrate through simulations the effectiveness of our method.

For the remainder of the paper, we present brief overviews of the DLMPC algorithm \cite{amoalonso2020distributed}, its explicit form \cite{amoalonso2020explicit}, and GPU parallelization in~\sec{dlmpc}. In~\sec{gpu} we analyze the GPU implementation and reduce the overhead via several concepts, such as longest-vector length, combined kernels and local memory. We demonstrate the usefulness of these improvements through simulations in~\sec{evaluation}.
Finally, we conclude the paper in~\sec{conclusion} along with some future research opportunities and directions.

\section{Preliminaries}\label{sec:dlmpc}
In this section we introduce the preliminary concepts necessary for the GPU implementation of DLMPC.
We start with a brief introduction of the DLMPC algorithm \cite{amoalonso2020distributed,amoalonso2020explicit}. We then follow with an overview of GPU parallelization.

\subsection{Distributed and Localized Model Predictive Control}

\begin{figure}
\includegraphics[,scale = 1.1]{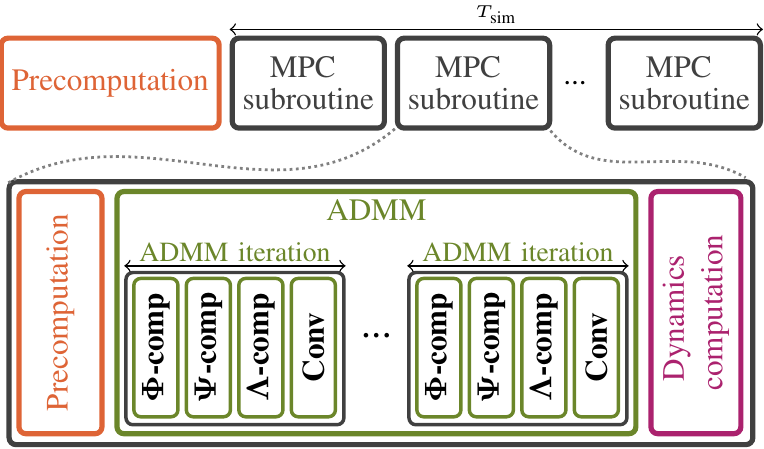}
\caption{The DLMPC algorithm consists of different computation steps. A precomputation step is carried out to compute necessary matrices that stay constant throughout all MPC iterations. Once this is completed, MPC iterations run sequentially (one per time-step) for a given number $\Tsim$ of time steps. Within a MPC iteration, a precomputation step precedes the ADMM algorithm. Once converged, we compute the next control input and the state according to the dynamics. Each iteration of ADMM is composed by the steps detailed in \alg{ADMM}.
}
\label{fig:baseline}
\end{figure}

\begin{figure}
\includegraphics[,scale = .75]{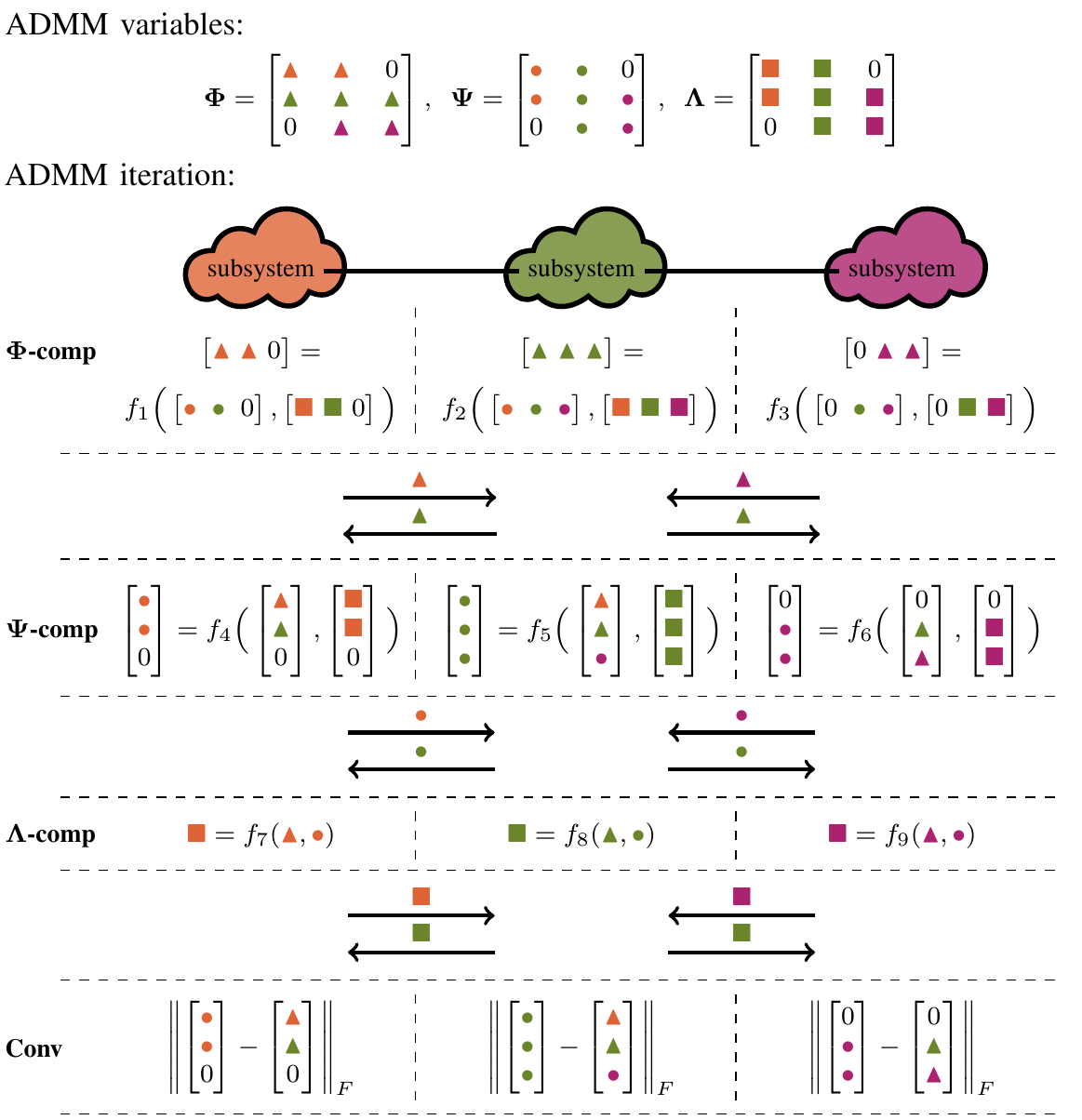}
\caption{An example showing how \alg{ADMM} distributes decision variables across a 3-node network.
In each iteration, {\Phicomp} distributes $\Phibf$ row-wisely, and then {\Psicomp} distributes $\Psibf$ column-wisely. After that, {\Lambdacomp} updates $\Lambdabf$ based on the new $\Phibf$ and $\Psibf$ and we determine if ADMM converges according to some criterions in the {\Conv} step.
Between each computation step, each subsystem shares the updated variable with its local neighborhood (in particular, the set of its $d$-hop neighbors).
}
\label{fig:ADMM}
\end{figure}

Consider a large-scale discrete-time linear time invariant (LTI) system with dynamic matrices $(A,B)$. The system is composed of $N$ subsystems interconnected according to an unweighted graph $\mathcal{G}_{(A,B)}$.
The control input for the system is computed through a MPC controller, where at each time step an optimal control problem is solved over the finite time horizon $T$ using the current state as initial condition for the dynamics. Hence, at time step $\tau$:
\OptMin[MPC]{
\sum_{t=0}^{T-1}x_{t}^{\mathsf{T}}Q_{t}x_{t}+u_{t}^{\mathsf{T}}R_{t}u_{t}+x_{T}^{\mathsf{T}}Q_{T}x_{T}
}{
\OptCons{
x_{0} = x(\tau),\, \ x_{t+1} = Ax_{t}+Bu_{t},
}{}{}\\
\OptCons{
x^{min}_t \leq x_t \leq x^{max}_t,\, x^{min}_T \leq x_T \leq x^{max}_T,
}{}{}\\
\OptCons{
u^{min}_t \leq u_t \leq u^{max}_t ~~ \forall t\in\{0,...,T-1\},
}{}{}\\
\OptCons{
\textit{$d$-locality constraints according to $\mathcal{G}_{(A,B)}$,}
}{}{}
}
where $x_t\in\mathbb R ^{N_x}$ is the state and $u_t \in\mathbb R^{N_u}$ the control input at time $t$ in the MPC problem. $x(\tau)$ represents the measured state at time step $\tau$, used as an initial condition for the MPC problem.
We require the states and inputs to be bounded above and below (marked by the superscript $max$ and $min$ respectively).

The $d$-locality constraints restrict the information exchange among subsystems to occur at a local scale,   i.e. subsystem $i$ can only share information with subsystems within distance $d \ll N$, as measured per the interconnection graph $\mathcal{G}_{(A,B)}$. Hence, the closed loop control policy for subcontroller $i$ can be computed using only states, control actions, and system models collected from $d$-hop neighbors. For a formal definition of this constraint, interested readers are referred to \cite{amoalonso2020distributed}.

In order to properly deal with the local communication constraints, we resort to the System Level Synthesis (SLS) parametrization \cite{anderson2019system,wang2019system}. In SLS, the decision variables of the MPC problem \eqn{MPC} are replaced by a matrix $\Phibf \in \mathbb R^{\big(N_xT+N_u(T-1)\big)\times N_x}$. Each row of $\Phibf$ corresponds to a state or an input in the system at a certain time step, so each state and input in the system occupy $T$ and $T-1$ rows of $\Phibf$, respectively. Locality constraints are translated into some suitable sparsity requirements on $\Phibf$.
In particular, the element of $\Phibf$ at the $i^\text{th}$ row and the $j^\text{th}$ column is non-zero only if the distance between the subsystems corresponding to the the $i^\text{th}$ row and the $j^\text{th}$ column is smaller than $d$ on $\mathcal{G}_{(A,B)}$.
Hence, the number of non-zero elements that each subsystem solves for is $O(d)$.
For a formal definition of locality constraints, see \cite{anderson2019system,wang2018separable}.

Once the MPC subroutine \eqn{MPC} is expressed in SLS, it is possible to distribute the optimization by means of the Alternating Direction Method of the Multipliers (ADMM) \cite{boyd2010distributed}. The ADMM algorithm consists of three steps. At each step, the required computations are distributed across the subsystems. Given the locality constraints, the dimension of the distributed subproblems will be dominated by $d$ as opposed to $N$.
\cite{amoalonso2020explicit} shows that all three steps can be efficiently computed:
Two steps admit a closed-form solution, and the other one has an explicit solution. Precomputations are necessary to enjoy the benefits of closed-forms and explicit solutions.
A sketch of the DLMPC algorithm is presented in \alg{DLMPC} and \alg{ADMM}, and we visualize the process in~\fig{baseline} and \fig{ADMM}.
\setlength{\textfloatsep}{0pt}
\begin{algorithm}[ht]
\caption{DLMPC algorithm sketch}\label{alg:DLMPC}
\begin{algorithmic}[1]
\State Precompute necessities of closed-form solutions.
\State Initialize $x(0)$.
\For {$t = 0$ {\bf to} $\Tsim$}
\State Precompute the explicit solution using $x(t)$.
\State Perform ADMM (see \alg{ADMM} for details).
\State Compute $u(t)$ from $\Phibf$ and obtain $x(t+1)$.
\EndFor
\end{algorithmic}
\end{algorithm}

\setlength{\textfloatsep}{0pt}
\begin{algorithm}[ht]
\caption{ADMM computations sketch for subsystem $i$}\label{alg:ADMM}
\begin{algorithmic}[1]
\State {conv} $\leftarrow$ {\bf false}.
\While{{conv} is {\bf false}}
\State \Phicomp: Compute rows of $\Phibf$ via the explicit solution in \cite{amoalonso2020explicit}. \label{alg:ADMM-row_comp}
\State Share $\Phibf$ with its $d$-hop local neighbors.
\State \Psicomp: Compute columns of $\Psibf$.
\State Share $\Psibf$ with its $d$-hop local neighbors.
\State \Lambdacomp: ${\Lambdabf} \leftarrow {\Lambdabf} + {\Phibf} - {\Psibf}$.
\State Share $\Lambdabf$ with its $d$-hop local neighbors.
\State \Conv: Check the convergence criterion and save the result in conv.
\EndWhile
\end{algorithmic}
\end{algorithm}

\subsection{GPU Parallelization Overview}

GPU differs from CPU in computation and memory, which have a profound impact on programming and implementations. We elaborate on those differences below.

\paragraph{Computing threads:}
A thread is the smallest independent sequence of instructions in a computing process. CPU is able to handle complex tasks using a limited number of threads in the order of $O(10)$. In contrast, GPU has the capacity of running thousands to millions of threads in parallel, but each is capable of simpler operations. A GPU computing process is referred as \emph{kernel}.

\paragraph{Sharing memory:}
Comparing to single-thread tasks in CPU, memory access is much more involved under a multi-thread scenario like GPU. When more than one thread access a sharing memory location, a race condition can occur when they both attempt to modify the content, and their access order determines the outcome. As a result, the consistency and correctness of the results are not guaranteed without special treatments. To avoid a race condition, we should either explicitly enforce some order among the threads or avoid concurrent access to the same memory locations. The former option is not preferred as it undermines the benefits of parallelism. On the other hand, memory sharing restriction curbs inter-thread information exchange.
As a result, an algorithm needs to avoid information exchange among its parallel components to achieve high performance on GPU.
\vspace*{0.25\baselineskip}

Given the characteristics, although GPU has great potential to boost algorithms' performance through parallelization, a GPU-parallelized algorithm is subject to two kinds of communication overheads: (i) \emph{CPU-GPU} and (ii) \emph{thread-to-thread}. CPU-GPU communication incurs an overhead on copying large volume of data between the memory systems of CPU and GPU, and thread-to-thread communication imposes an overhead on handling coupled memory accesses among parallel threads.
GPU local memory\footnote{We use OpenCL terminology, also referred as \emph{shared memory} in CUDA.} offers a limited resource to mitigate these overheads. It provides a narrow memory space accessible to a \emph{local} group of threads with synchronization barriers, and hence allows local groups of threads to exchange information through these shared memory locations with consistency guarantee.
Meanwhile, local memory is limited to a local scale, and a single thread cannot belong to two different groups.

\section{GPU Parallelization}\label{sec:gpu}

In this section, we provide the implementation of the DLMPC algorithm in a computer system equipped with both a CPU and a GPU. First, we take advantage of the parallelization potential of the algorithm and perform a \emph{naive parallelization} of the ADMM steps in GPU. We then provide enhancements to reduce overheads, such as reducing setup complexity by using \emph{longest-vector length}, and reducing CPU-GPU communication by setting \emph{combined kernels}. Lastly, we present how the locality constraints allow to effectively use \emph{local memory} to deal with the thread-to-thread coupling.

\subsection{Naive Parallelization}
\begin{figure}
\includegraphics[scale = 1.1]{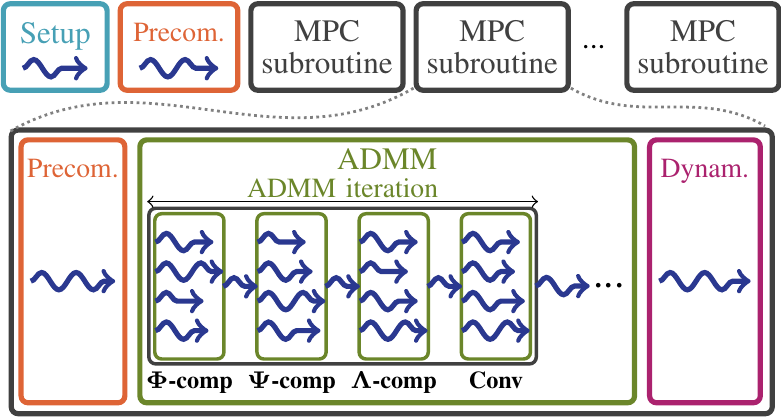}
\caption{Naive parallelization implementation. Boxed components represent the same computation steps as in~\fig{baseline}. An additional Setup step has been introduced at the beginning of the implementation for GPU setup. Computing threads are denoted with blue wavy arrows: A single arrow represents a single-thread computation, and multiple arrows within a computation step represent a multi-thread computation. The length of the different arrows in multi-thread processes represents runtime for each of the threads.}
\label{fig:naive_gpu}
\end{figure}

We start with a naive parallelization of DLMPC algorithm by parallelizing ADMM steps in \alg{ADMM}.
For each ADMM step, we assign each of the subproblems below a single thread in GPU:
\begin{itemize}
\item For \Phicomp, each thread computes one row of $\Phibf$.
\item For \Psicomp, each thread computes one column of $\Psibf$.
\item For \Lambdacomp, each thread computes one element of $\Lambdabf$.
\item For \Conv, each thread evaluates the convergence criterion against one column.
\end{itemize}

Notice that each thread in GPU is tasked with performing all necessary arithmetic operations leading to the assigned row/column/element. In this implementation we are not parallelizing the arithmetic computations in GPU, but rather treating each GPU thread as a subsystem of the distributed MPC framework.

According to the ADMM algorithm, abundant information sharing is required after each computation in order to perform the next computation. Due to the limitations of GPU communication among threads, in this naive parallelization we perform the information sharing in the form of memory accesses in CPU. Therefore, after each parallelized computation we return to CPU to exchange results and set up the next one.
We illustrate this implementation in~\fig{naive_gpu}, where we represent the computing threads with an arrow so one can distinguish the steps that are computed in CPU (single thread) and the ones that are computed in GPU (multi-thread). Notice that an additional setup step required to launch the GPU kernels is also represented.

This naive parallelization of the DLMPC algorithm suffers from several overheads.
First, the threads have different runtime due to the various lengths of the rows and columns they compute. Such various lengths result in significant setup overhead before computation.
Second, there are several CPU-GPU switches per ADMM iteration to exchange information across different threads for rows and columns. This incurs CPU-GPU communication overhead.
Those overheads imply that a naive parallelization of a distributed and localized MPC scheme such as DLMPC is not necessarily efficient, and additional considerations are needed to fully exploit the GPU potential. In what follows, we analyze these overheads and provide effective solutions based on hardware-specific considerations and the presence of locality constraints. We build upon these solutions until an optimal GPU implementation of the ADMM steps is presented.

\subsection{Longest-Vector Length}

\begin{figure}
\centering
\includegraphics[,scale = 1.4]{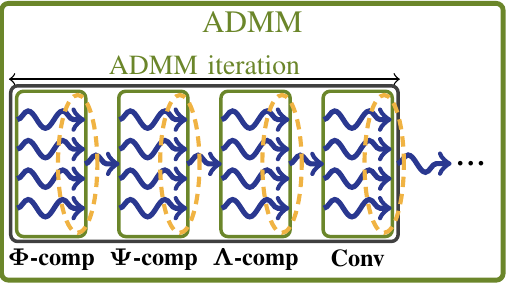}
\caption{Since the computation time of parallelized threads is dominated by the slowest one, we could omit computing the exact input vector size for each thread and feed all threads with same-sized (the longest-sized) input vectors, which results in roughly the same runtime for all the threads. This simplifies both the setup and per-thread computation and hence reduces the overhead.
}
\label{fig:worst-case_length}
\end{figure}

Threads have different runtime in~\fig{naive_gpu} since they process different sizes of input vectors. Feeding each thread a different-sized input vector imposes a setup overhead -- we need to compute, store, and pass as parameters of the threads the sizes of each input vector. Such an overhead was justifiable in a single thread CPU version like \cite{li2019slstoolbox} as we want to avoid unnecessary computations. In particular, only non-zero elements are needed when computing on a single thread, and filtering out non-zero elements pays off as fewer inputs imply faster computation under sequential processing. The situation changes in GPU parallelization. Since the computation time of parallelized threads is determined by the slowest one, and the kernel does not return until \emph{all} threads have finished the computations, it is no longer beneficial to trim off zero elements unless they are processed by the slowest thread.

Accordingly, we can save the efforts of attaining exact different-sized input vectors for parallelization. Instead, we only need to ensure the input vector is long enough to cover the non-zero elements and find the minimum upper bound on the length, which is the maximum number of non-zero elements in the rows and columns of $\Phibf$ and $\Psibf$, respectively, or the \emph{longest-vector length} for short.
We denote by $\Drow$ and $\Dcol$ the longest-vector lengths of $\Phibf$ and $\Psibf$, respectively, and establish below that by virtue of the locality constraints, $\Drow,\Dcol\ll N$ and the number of elements that a thread solves for is much smaller than the size of the network.

\begin{lemma}
Let $s$ be the maximum number of states or control inputs per subsystem in the network, and $l$ the maximum degree of nodes in $\mathcal G_{(A,B)}$. Suppose $\mathcal G_{(A,B)}$ is subject to $d$-locality constraints and the MPC time horizon is $T$, then $\Drow$ and $\Dcol$ are bounded by
\begin{align*}
\Drow \leq \frac{s(l^{d+1} - 1)}{l - 1}, \quad\quad
\Dcol \leq \frac{(2T-1)s(l^{d+1} - 1)}{l - 1}.
\end{align*}
\end{lemma}

\begin{proof}
Each row in $\Phibf$ represents a state/input in a subsystem, and hence $\Drow$ is the number of states that it can receive information from. We can establish a bound on $\Drow$ by bounding the number of nodes within $d$-hops. By definition, we know that there are at most $l^k$ nodes that are $k$-hops away from a node, so within $d$-hops, there are at most
\begin{align*}
1 + l + l^2 + \cdots + l^d = \frac{l^{d+1} - 1}{l - 1}
\end{align*}
nodes, each has at most $s$ states, which shows the bound.

On the other hand, $\Dcol$ is the number of states and inputs, among all horizon $T$, a state can impact. Similarly, by $d$-locality constraints, we can use the bound on $\Drow$ as an estimate of the states/inputs a state would impact at each time. Since there are, in total, $T$ states and $T-1$ inputs in $\Psibf$ per column, we can bound $\Dcol$ by $(2T - 1)$ times of the above bound on $\Drow$, which yields the desired result.
\end{proof}

Given a sparse system, the maximum degree $l$ is expected to be small, as is the maximum number of subsystem states/inputs $s$. Given the assumption $d \ll N$, the above lemma suggests $\Drow,\Dcol\ll N$

\subsection{Combined Kernels}

\begin{figure}
\centering
\includegraphics[,scale = 1.4]{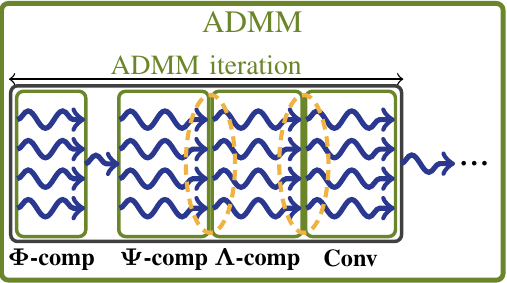}
\caption{To reduce CPU-GPU communication overhead, we remove the single-thread computation in between computations (circled in yellow) by combining the computations into per-column threads. This enhances the parallelization by combining kernels. Such a combination does not apply to the information exchange between {\Phicomp} and {\Psicomp} as we shift from per-row to per-column computation.}
\label{fig:combined_kernels}
\end{figure}

We then tackle the CPU-GPU communication overhead. In the naive parallelization, CPU-GPU communication are necessary to properly exchange information in between computations. The reason is that each computation occurs according to a different distribution of the elements of $\Phibf$, $\Psibf$, $\Lambdabf$, i.e. row-wise, column-wise, and element-wise. Although this particular distribution of elements might be the most efficient for each of the computations isolated, the additional GPU-CPU overhead steaming from the information sharing in between computations makes this option is suboptimal.

To reduce the CPU-GPU communication overhead, we proposed the use of combined kernels. In particular, the last three computation steps in the ADMM iteration can be combined in the same kernel by parallelizing {\Lambdacomp} in a column-wise fashion (as opposed to element-wise). By distributing the computations in this manner, each of the threads has sufficient information to sequentially perform {\Psicomp}, {\Lambdacomp}, and {\Conv} without the need for communication among threads. This reduces the CPU-GPU communication overhead since only one exchange between CPU and GPU is necessary for each ADMM iteration (for the transformation from row-wise to column-wise). However, the treatment herein slightly degrades the parallelization benefits of {\Lambdacomp},
since by having only a thread per column, each thread now has to loop over the elements in its column sequentially. This additional overhead is very modest because due to the locality constraints, the number of relevant elements per thread is $D_\text{col}$, as opposed to a CPU-GPU memory-copying operation, where the variables handled are of order $N$.

\subsection{Local Memory and Column Patch}
\begin{figure}
\centering
\includegraphics[,scale = 1.4]{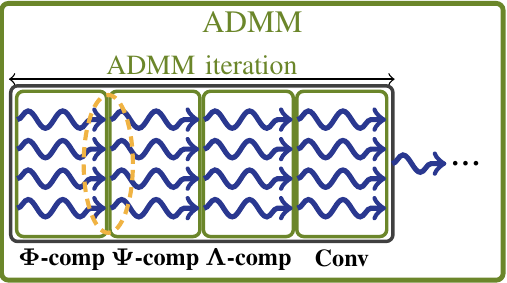}
\caption{Conducting {\Phicomp} and {\Psicomp} consecutively in GPU requires the per-row threads to exchange information with per-column threads, which results in thread-to-thread communication. Using local memory, we can avoid such a thread-to-thread communication by creating column patches, which duplicate row-wise computation threads to decouple the per-column local memory groups and synchronize results entirely within GPU.
}
\label{fig:local_memory}
\end{figure}

The information exchange between {\Phicomp} and {\Psicomp} involves the transformation from row-wise to column-wise representation and hence is not easily combined into one kernel. The key difficulty is that the row-wise results should be passed down to per-column threads, which results in thread-to-thread communication/synchronization issues. We could realize synchronized thread-to-thread communications through local memory. However, local memory is shallow -- It would not fit all the threads in -- and it is exclusive -- One thread can only belong to one group. These properties make the bipartite information exchange pattern difficult to enforce: The row-wise result might be required by multiple per-column threads, while each thread may need multiple row-wise results. As a result, to leverage local memory to save CPU-GPU communication overhead, we need to group the threads smartly.

Our approach is to group each per-column thread in {\Psicomp} with row-wise computation threads in {\Phicomp}, referred to as \emph{column patch}. That is, for the $i^{\text{th}}$ column of $\Psibf$, we launch a column patch to solve for the $j^{\text{th}}$ row where $j\in\{j:\ \Phibf(j,i)\neq0\}$. Once this is done, the row-wise results are saved in local memory in GPU, and one of those threads can proceed with the column-wise computations of $\Psibf$ and $\Lambdabf$ as described in the previous subsection without returning to CPU.

Thanks to the locality constraints, each column patch only has $\Dcol$ {\Phicomp} threads to include and fit their results in the shallow local memory. On the other hand, since each row has several non-zero elements, we could potentially have multiple threads in different column patches that compute the same row-wise result. But it is fine as GPU has plenty of threads to launch, and using multiple threads to compute the same result in parallel does not incur additional runtime overhead.
Therefore, we ensure synchronization without the need for information sharing across threads - since we can repeat relevant computations in different local groups - or computing units - since local synchronization is all is needed.
This was only possible by exploiting the GPU architecture together with the locality constraints directly encoded in the DLMPC formulation.

We highlight the roles of the locality constraints in our enhancement techniques. Locality constraints can facilitate desirable trade-offs between computational resources and information exchange across threads:
Longest-vector length incurs additional precomputation steps, and combined kernels sacrifice parallelization potential of {\Lambdacomp}.
Those trade-offs are justified by the small dimensions derived from the locality constraints $D_\text{row},D_\text{col}\ll N$.
Meanwhile, for the local memory and column patch technique, locality constraints allow us to decouple the row-wise threads without harming the runtime (in essence, locality allow us to pay in the spatial space to decouple the row-wise threads without temporal performance degradation).

\section{Evaluation}\label{sec:evaluation}

Through simulations, we study two aspects of the GPU-parallelized DLMPC. First, we compare the scalability of our implementation with other methods. We then analyze the overhead of the implementations for future enhancements.

\subsection{Setup}

We implement our GPU-parallelized DLMPC in Python and OpenCL.

We compare the scalability of the four proposed GPU implementations in~\sec{gpu} against two CPU variations -- a Python replica of the single-threaded DLMPC version in~\cite{li2019slstoolbox} and an optimization-based approach under the SLSpy framework \cite{tseng2020slspy} with CVXPY \cite{cvxpy} as the solver.
The results are measured on a desktop with AMD Ryzen 7 3700X processor (16 logical cores), 32 GB DDR4 memory, and AMD Radeon RX 550/550X GPU.
For each evaluated scenario, we simulate $100$ different initial conditions and present the average and the standard deviation of the measurements.
The synthetic dynamics is chosen the same as in~\cite{amoalonso2020explicit}, which is
a chain-like network with two-state nodes as the subsystems. Each subsystem $i$ evolves according to
\begin{equation*}
[x(t+1)]_{i}=[A]_{ii}[x(t)]_{i}+\sum_{j\in\textbf{in}_{i}(d)}[A]_{ij}[x(t)]_{j}+[B]_{ii}[u(t)]_{i},
\end{equation*}
where $\textbf{in}_{i}(d)$ contain the $d$-hop neighbors of node $i$ and
\begin{equation*}
[A]_{ii}=\begin{bmatrix}
1  & 0.1 \\
-0.3 &  0.7
\end{bmatrix}, \ [A]_{ij}=\begin{bmatrix}
0  & 0 \\
0.1 &  0.1
\end{bmatrix}, \ [B]_{ii}=I.
\end{equation*}
The state is subject to upper and lower bounds:
\begin{equation*}
-0.2\le[x(t)]_{i,1} \le1.2 \quad  \text{for } t=1,...,T,
\end{equation*}
where $[x]_{i,1}$ is the first state in the two-state subsystem $i$.

We perform the MPC with $\Tsim = 20$ subroutine iterations with the cost function $$f(x,u) = \sum_{i=1}^{N} \sum_{t=1}^{T-1} \|[x(t)]_i\|_2^2+\|[u(t)]_i\|_2^2+\|[x(T)]_i\|_2^2.$$ Note that the number of states and inputs in this plant is $3N$, since each of the $N$ subsystems has $2$ states and $1$ input.

\subsection{Scalability}

\begin{figure}
\centering
\includegraphics{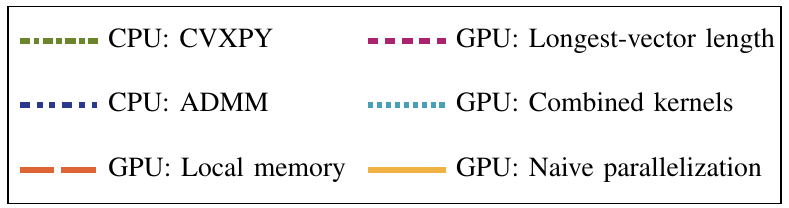}
\includegraphics{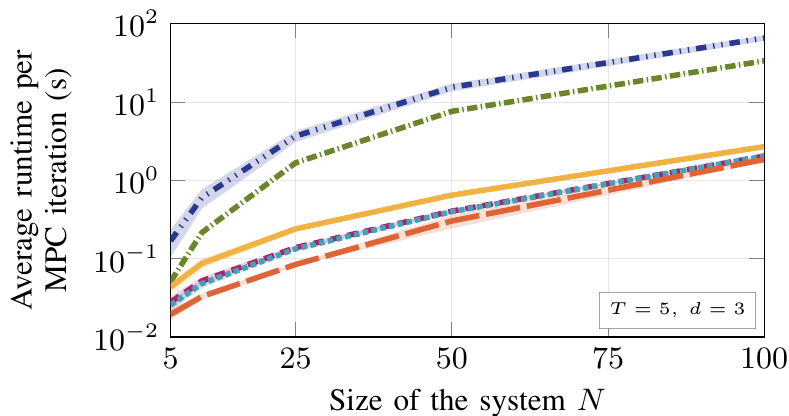}
\\[0.5\baselineskip]
\includegraphics{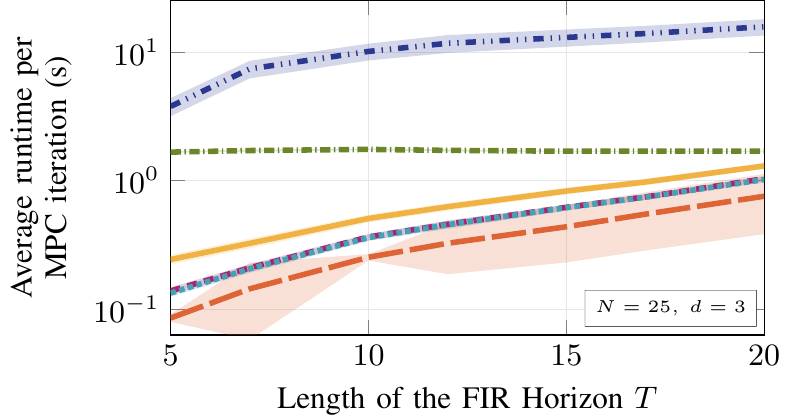}
\\[0.5\baselineskip]
\includegraphics{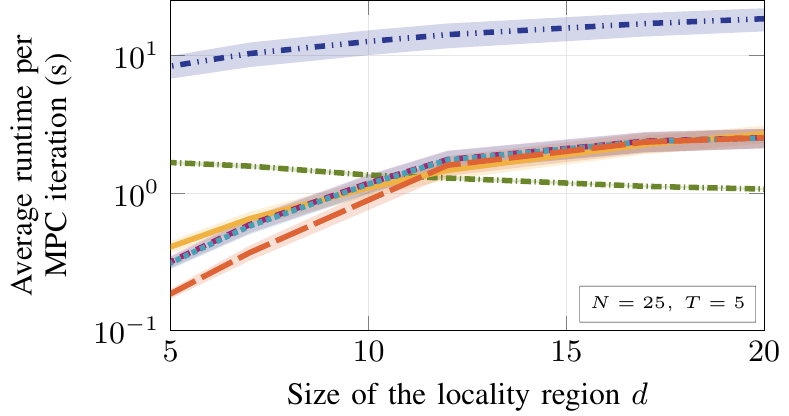}
\caption{Comparison of the runtimes obtained by different computing strategies for different parameter regimes. The lines are the mean values and the shaded areas show the values within one standard deviation.
We observe that GPU computation strategies scale much better with the size of the network than CPU implementations.
This trend also holds for ADMM CPU implementation over all time horizons and locality region sizes, while GPU implementations only outperform CVXPY on small locality regions and CVXPY seems to scale well over all simulated time horizons.}
\label{fig:scalability}
\end{figure}

To evaluate the scalability of the methods, we run the simulations with varying system size $N$, MPC time horizon $T$, and locality region size $d$.
We measure the average runtime per MPC iteration, i.e., the total runtime divided by the number of MPC subroutine iterations $\Tsim$, and summarize the results in~\fig{scalability}.
We remark that the runtime is measured for the \emph{whole} simulation rather than merely the ADMM portion of the algorithm.

In~\fig{scalability}, the GPU implementations scale much better with the network size $N$ than the CPU implementations. Moreover, the runtime differences grow from an order of magnitude to several orders of magnitude as $N$ increases. This is as expected since
GPU implementations can parallelize the subsystem computations through multi-threads whereas the CPU implementations cannot.
Remarkably, the ADMM implementation in CPU is consistently worse than when solved via CVXPY, which emphasizes the need for a parallel implementation such as the one presented in this paper to fully take advantage of the DLMPC algorithm. Among the GPU methods, local memory has superior performance ($50\times$ faster than ADMM in CPU for $N=25$ and $35\times$ for $N=100$), followed by combined kernels and longest-vector length ($15\times$ faster than ADMM in CPU for $N=25$ and $25\times$ for $N=100$). Naive parallelization is $3\times$ slower than local memory for small $N$. In fact, a smaller $N$ leads to a bigger performance difference among GPU implementations, which indicates that the CPU overhead of the implementations outweighs the improvements made by GPU when the network scales.

For time horizon $T$ and locality region size $d$, the runtime scales accordingly for all the ADMM implementations. This can be seen from a simple analysis of the optimization variables:
Since larger $T$ and $d$ introduce more non-zero elements in the matrix $\Phibf$, the corresponding decomposed row and column vectors become longer and the runtime increases.
However, this trend does not apply to the SLSpy-based CPU implementation solved by CVXPY, where the runtime stays pretty much the same or even decreases when $T$ and $d$ increases. Nevertheless, the GPU implementations still outperform CVXPY by up to one order of magnitude for $d \leq 10$ and all simulated $T$. In other words, the advantage of GPU parallelization is significant when the locality region is small. Hence, in light of the results, the ADMM implementation of DLMPC is best for large network consisting of subsystems with relatively small local neighborhood.

\subsection{Overhead Analysis}
\begin{figure}
\centering
\vspace*{0.25\baselineskip}
\includegraphics{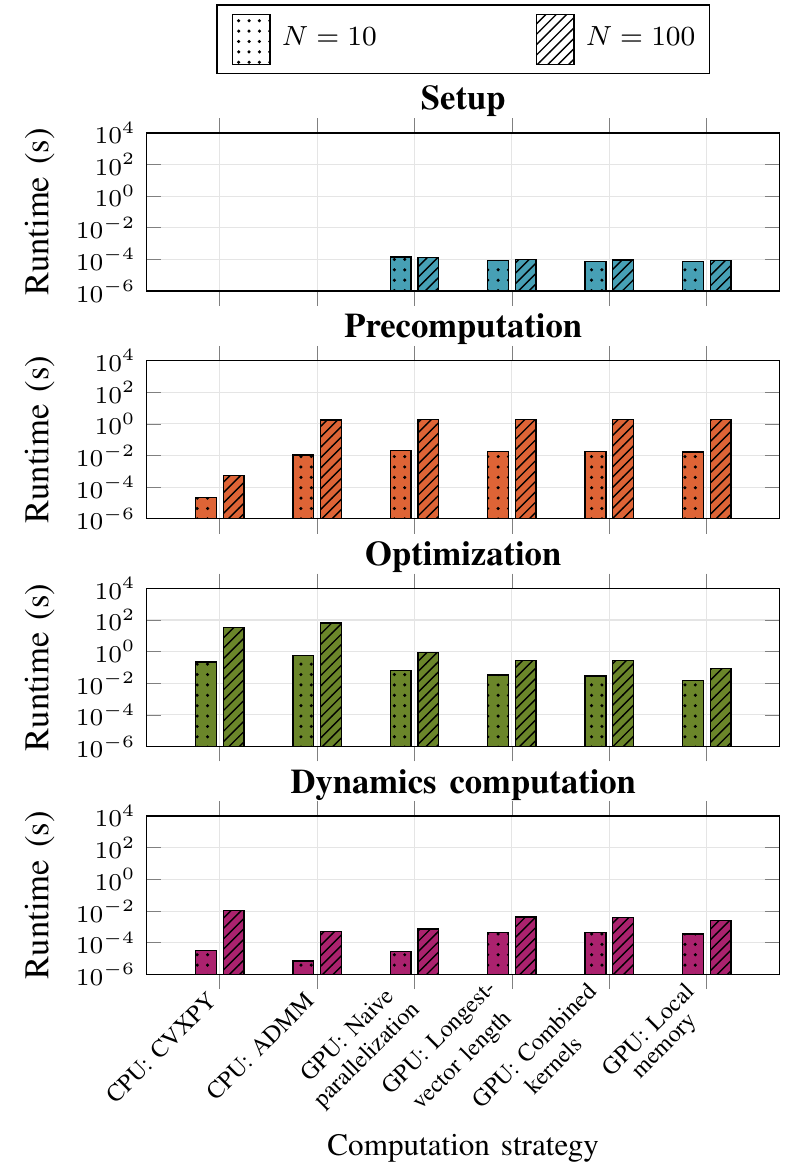}
\caption{Runtime breakdown for the different computing phases: precompilation, precomputation, optimization and dynamics computation for two different network sizes $N=10$ and $N = 100$. Colors for the different phases are the same as in \fig{naive_gpu}. For CPU implementations, the optimization phase is the bottleneck, while the bottleneck shifts to the precomputation for the GPU implementations.}
\label{fig:breakdown}
\end{figure}

To better understand how different methods scale in runtime with the network size, we break down the runtime into four DLMPC phases as in \fig{naive_gpu}, which are

\paragraph{Setup} refers to all the computational steps necessary to setup the GPU computations. It refers to the compilation of the kernels themselves, but also to all the extra computations necessary to implement the different GPU schemes, such as the computation of $\Drow$ and $\Dcol$, etc.

\paragraph{Precomputation} refers to computing the necessary matrices and vectors that multiply the decision variables $\Phibf$, $\Psibf$, $\Lambdabf$. We note here that this phase appears twice in the implementation: before starting any of the MPC computations, and before starting each of the MPC subroutines.

\paragraph{Optimization} refers to the computational steps necessary to solve the optimization problem \eqn{MPC}. In this case, it refers to the steps taken by ADMM or CVXPY to find a value for $\Phibf$.

\paragraph{Dynamics computation} refers to the computation of the next state using the control input given by MPC.
\vspace*{0.25\baselineskip}

We measure the runtime breakdown for network size $N = 10$ and $N = 100$ and present the results in~\fig{breakdown}.
From \fig{breakdown}, we can observe that the setup and dynamics computation phases are relatively fast, and the runtime is dominated by the optimization phase when $N = 10$ and the precomputation phase when $N = 100$. As we adopt more sophisticated techniques described in \sec{gpu} to GPU implementations, the optimization phase speeds up significantly. However, the precomputation phase does not benefit from the techniques and it becomes the bottleneck when the network size scales.
Since an increase in $d$ or $T$ mostly burdens the setup and precomputation phases, the results also explain why the GPU implementations scale poorer when $d$ or $T$ is large.

Another important takeway from \fig{breakdown} is that for CPU implementations, the optimization phase is the bottleneck, while the bottleneck shifts to the precomputation for the GPU implementations. This justifies our approaches in this paper to accelerate the CPU optimization phase by GPU parallelization. Meanwhile, the results also suggest that to further speed up the computation, future research should focus on faster precomputation techniques.

\section{Conclusion and Future Directions}\label{sec:conclusion}

We develop effective GPU parallelized DLMPC for large-scale distributed system control. Our results show that although a naive GPU implementation does improve the performance by $15-25 \times$, we can still get up to $50 \times$ performance improvement by taking into account the hardware-intrinsic limitations. We overcome these limitations by taking advantage of the local communication constraints in the formulation, and developing longest-vector length, combined kernels and local memory implementations. With extensive experiments, we demonstrate that the DLMPC algorithm is suitable for GPU parallelization, and that its full potential is only realized when the local communication constraints are taken into account in the GPU implementation. We demonstrate the scalability of the method for large network sizes, and noticed that most of the computational overhead in the GPU computations was due to the precomputations being performed in CPU.

As discussed in the overhead analysis, precomputation becomes the new overhead after our GPU parallelizations. Therefore, a future direction would be to effectively parallelize the precomputation for higher performance.
In addition, there are some other parts of DLMPC that we can parallelize further, such as a
better initial point for ADMM to converge faster as well as a better parallelized dynamics computation.
Another future direction is to include more than one ADMM iteration into one kernel, to avoid CPU-GPU data exchanges.
One might also be interested in extending the parallelization in this work to a fully distributed setting where the processing units scatter over a network. As the distributed setting introduces new challenges such as robustness to communication dropouts, synchronization, and delay, it would be interesting to see how locality constraints could improve robustness or performance.

\bibliographystyle{IEEEtran}
\bibliography{Test}

\end{document}